\documentclass[11pt]{article}
\usepackage[margin=1in]{geometry}
\usepackage{fullpage,wrapfig}
\usepackage{amsmath, amsthm, amssymb, algorithm,thmtools,algpseudocode,enumerate,caption,framed}
\usepackage{paralist, sidecap}
\usepackage[usenames,dvipsnames,svgnames,table]{xcolor}
\definecolor{darkgreen}{rgb}{0.0,0,0.9}
\usepackage[colorlinks=true,pdfpagemode=UseNone,citecolor=Blue,linkcolor=Red,urlcolor=Black,pagebackref]{hyperref}
\usepackage{tikz}
\usepackage[T1]{fontenc}
\usepackage[utf8]{inputenc}
\usepackage{authblk}
\usepackage{mdframed}
\usepackage{multirow,array}
\usepackage{diagbox}
\usepackage{graphicx,xcolor}
\usepackage{subcaption}
\usepackage[shortlabels]{enumitem}

\newtheorem{theorem}{Theorem}[section]

\newtheorem{lemma}{Lemma}[section]

\newtheorem{obs}{Observation}[section]

\newcommand{\ugrid}[2]{\underline{\Gamma}}

\newcommand{\opt}{S^*}
\newcommand{\brc}[1]{\text{brc}(#1)}
\newcommand{\bs}[1]{\text{bs}(#1)}
\newcommand{\ls}[1]{\text{ls}(#1)}
\newcommand{\ts}[1]{\text{ts}(#1)}
\newcommand{\rs}[1]{\text{rs}(#1)}

\newcommand{\grid}[2]{\Gamma}
\newcommand{\pgraph}[1]{G(#1)}

\title{Packing Boundary-Anchored Rectangles and Squares\thanks{A preliminary version of this paper appeared in the proceedings of 29th Canadian Conference on Computational Geometry (CCCG 2017)~\cite{BiedBMM17}. The work of TB and AM is supported in part by Natural Sciences and Engineering Research Council of Canada (NSERC).}}

\author[1]{Therese Biedl}
\author[1]{Ahmad Biniaz}
\author[2]{Anil Maheshwari}
\author[2]{Saeed Mehrabi}

\affil[1]{{\small David R. Cheriton School of Computer Science, University of Waterloo, Waterloo, Canada.

\texttt{biedl@uwaterloo.ca, ahmad.biniaz@gmail.com}}
}
\affil[2]{{\small School of Computer Science, Carleton University, Ottawa, Canada.

\texttt{anil@scs.carleton.ca, saeed.mehrabi@carleton.ca}}
}

\date{}

\begin{document}

\maketitle

\begin{abstract}
Consider a set $P$ of $n$ points on the boundary of an axis-aligned square $Q$. We study the \emph{boundary-anchored packing} problem on $P$ in which the goal is to find a set of interior-disjoint axis-aligned rectangles in $Q$ such that each rectangle is \emph{anchored} (has a corner at some point in $P$), each point in $P$ is used to anchor at most one rectangle, and the total area of the rectangles is maximized. Here, a rectangle is anchored at a point $p$ in $P$ if one of its corners coincides with $p$. In this paper, we show how to solve this problem in time linear in $n$, provided that the points of $P$ are given in sorted order along the boundary of $Q$. We also consider the problem for anchoring \emph{squares} and give an $O(n^4)$-time algorithm when the points in $P$ lie on two opposite sides of $Q$.
\end{abstract}

\section{Introduction}
\label{sec:introduction}
Let $Q$ be an axis-aligned square in the plane, and let $P$ be a set of points in $Q$. Call a rectangle $r$ {\em anchored} at a point $p\in P$ if $p$ is a corner of $r$. The \emph{anchored rectangle packing} (ARP) problem is to find a set $S$ of interior-disjoint axis-aligned rectangles in $Q$ such that each rectangle in $S$ is anchored at some point in $P$, each point in $P$ is a corner of at most one rectangle in $S$, and the total area of the rectangles in $S$ is maximized; see Figure~\ref{problem-fig}(a). It is not known whether this problem is \textsc{NP}-hard. The best known approximation algorithm for this problem achieves ratio $7/12-\varepsilon$ due to Balas et al.~\cite{BalasDT16}, who also studied several variants of this problem.

In this paper, we study a variant of the anchored packing problem in which all the points of $P$ lie on the boundary of $Q$. We refer to this variant as the \emph{boundary-anchored rectangle packing (BARP)} problem when the anchored objects are rectangles (see Figure~\ref{problem-fig}(b)), while when we require to anchor \emph{squares} instead of rectangles, we call the problem the \emph{boundary-anchored square packing (BASP)} problem. We first present an algorithm that solves the BARP problem in linear time, provided that the points of $P$ are given in sorted order along the boundary of $Q$ (Section~\ref{algorithm-section}). Despite the simplicity of our algorithm, its correctness proof is non-trivial (Section~\ref{correctness-section}). Then, we consider the BASP problem and give an $O(n^4)$ algorithm for this problem when the points in $P$ are on two opposite sides of $Q$ (Section~\ref{sec:squares}).

\begin{figure}[t]
	\centering
	\setlength{\tabcolsep}{0in}
	$\begin{tabular}{cc}
	\multicolumn{1}{m{.5\columnwidth}}{\centering\includegraphics[width=.38\columnwidth]{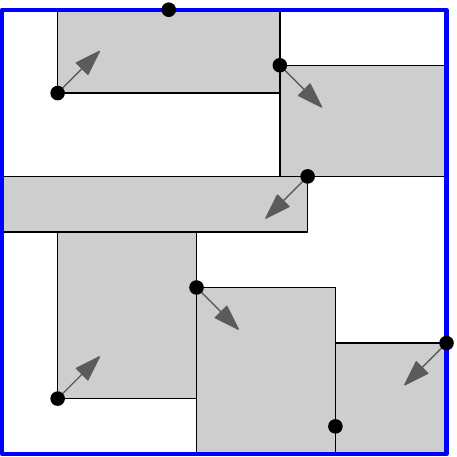}}
	&\multicolumn{1}{m{.5\columnwidth}}{\centering\includegraphics[width=.38\columnwidth]{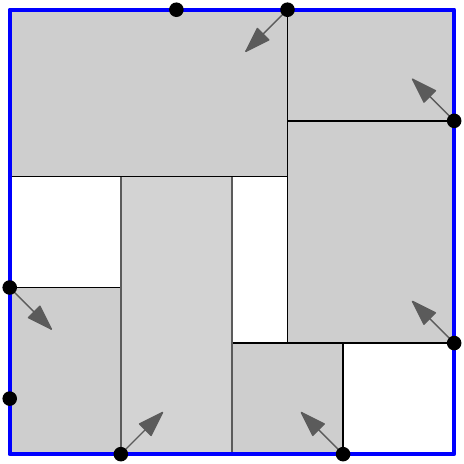}}
	\\
	(a)&(b)
	\end{tabular}$
	\caption{Instances and (non-optimal) solutions of (a) the ARP problem, and (b) the BARP problem.}
	\label{problem-fig}
\end{figure}

\paragraph{Related results.} The rectangle packing problem is related to strip packing and bin packing problems, which are well-known optimization problems in computational geometry. Rectangle packing problems have applications in map labeling~\cite{KakoulisT13, KreveldSW99}. Balas et al.~\cite{BalasDT16} studied several variants of the anchored packing problem; namely, the {\em lower-left anchored rectangle packing} problem in which points of $P$ are required to be on the lower-left corners of the rectangles in $R$, the {\em anchored square packing} problem in which every anchored rectangles is required to be a square, and the {\em lower-left anchored square packing} problem which is a combination the two previous problems. For the lower-left rectangle packing problem, Freedman~\cite{Tutte69} conjectured that there is a solution that covers $50\%$ of the area of $Q$. The best known lower bound of $9.1\%$ of the area of $Q$ is due to Dumitrescu and T\'{o}th~\cite{DumitrescuT15}. Balas et al.~\cite{BalasDT16} presented approximation algorithms with ratios $(7/12-\varepsilon)$ and $5/32$ for anchored rectangles and anchored square, respectively. They also presented a 1/3-approximation algorithm for the lower-left anchored square packing problem, and proved that the analysis of the approximation factor is tight. Balas and T\'{o}th~\cite{BalasT16} studied the combinatorial structure of maximal anchored rectangle packings and showed that the number of such distinct packings with the maximum area can be exponential in the number $n$ of points of $P$; they give an exponential upper bound of $2^n C_n$, where $C_n$ denotes the $n$th Catalan number.

Finally, Iturriaga and Lubiw~\cite{IturriagaL03} studied the \emph{elastic labeling} in which we are given a set of points on the boundary of a rectangle and the goal is to anchor rectangles to these points such that no two of them overlap. Here, elastic means that the label can have varying width and height, but each rectangle has a pre-specified fixed area. This problem is different than BARP as for the objectives of the two problems for instance, or having no pre-specified restriction on the area of the anchored rectangles in the latter.

\section{Boundary-Anchored Rectangles}
\label{algorithm-section}
In this section, we give a linear-time algorithm for the BARP problem. Before describing the algorithm, we first briefly argue that BARP is solvable in polynomial time.

\paragraph{An outline.} It is easy to see~\cite{BalasDT16} that in any rectangle packing the boundaries of rectangles must lie on the {\em grid} $\grid{P}{Q}$: the union of the boundary of $Q$ and the extension of inward rays from all points until they hit the opposite boundary. For each point $p\in P$, there are $O(n^2)$ potential rectangles of $\grid{P}{Q}$ anchored at $p$, and so we have a total of $O(n^3)$ candidate rectangles from which we must pick an independent set (with respect to their intersection graph) such that the sum of the weights (defined to be the area of each rectangle) is maximized. If all points are on the boundary, then it is easy to represent each rectangle as a {\em string} (i.e., a Jordan curve) such that all strings have a point on the infinite face and two strings intersect if and only if not both rectangles should be taken; see Figure~\ref{fig:onePointStrings}. This class of graphs is known as the \emph{outer-string graphs} for which it is known that maximum-weight independent set is solvable in $O(N^3)$ time, where $N$ denotes the number of segments in a geometric representation of the input graph~\cite{KeilMPV17}. As such, BARP is solvable in $O(n^9)$ time, but this is rather slow.

\begin{figure}[t]
	\centering
	\includegraphics[width=.40\textwidth]{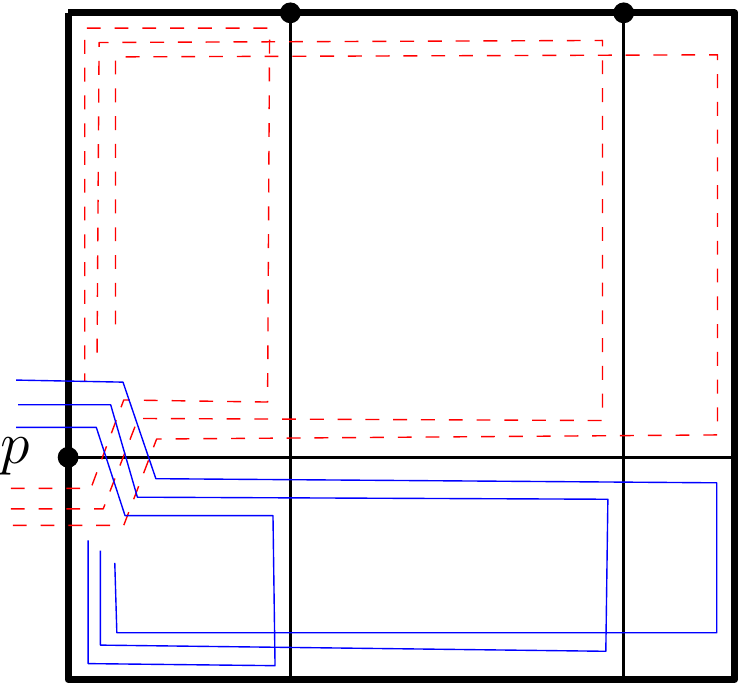}
	\caption{BARP can be solved via maximum-weight independent set in an outer-string graph.}
	\label{fig:onePointStrings}
\end{figure}

In this section, we give key insights that lead to a much faster algorithm. Define a {\em cell} to be a maximal rectangle not intersected by lines of the grid $\grid{P}{Q}$. Given an optimum solution $S$, define a {\em hole} of $S$ to be a maximal connected region of $Q$ that is not covered by $S$, see Figure~\ref{grid-fig}(b). We show the following in Section~\ref{correctness-section}:
\begin{theorem}
\label{ins:standardFormOtimal}
An optimal solution $S$ either covers all of $Q$, or it has exactly one hole which is a single cell.
\end{theorem}

It is quite easy to test whether all of $Q$ can be covered (see Lemma~\ref{entire-Q-lemma}). In particular, this is always feasible if two points have the same $x$-coordinate or the same $y$-coordinate. In consequence, for the following discussion we assume that no two points have the same $x$- or $y$-coordinate. If we cannot cover all of $Q$, then we want to minimize the size of the hole. However, there are a quadratic number of cells, and more crucially, not all cells are feasible; i.e., could be holes. The second key result is therefore the following (by Theorem~\ref{alg-thr}): 
\begin{lemma}
\label{lem:testForCoverACell}
Given the (at most four) points defining the boundary of a cell $\psi$, we can test in $O(1)$ time whether some packing covers $Q-\psi$.
\end{lemma}

This immediately gives an $O(n^2)$ algorithm to find the best solution of type $Q-\psi$: consider the cells in order, test whether they are feasible and then find the corresponding packing that maximizes the area among those that are feasible. However, it is not necessary to test each cell individually. We can characterize exactly when a cell $\psi$ is feasible, based solely on where the supporting lines of $\psi$ (which are either the boundary of $Q$ or rays emanating from some points) have their endpoints. Hence, we do not need to look at individual cells, but at the list of points on the four sides, to find the minimum area hole. In the following, we describe this in more details.

We write $P_\mathcal{B}$ (resp., $P_\mathcal{L}, P_\mathcal{T}$ and $P_\mathcal{R}$) for the points of $P$ on the bottom (resp., left, top and right) side. For a point $p$ in the plane, we denote by $x(p)$ and $y(p)$ the $x$- and $y$-coordinates of $p$, respectively. The following theorem, which we will prove in Section~\ref{correctness-section}, characterizes possible optimal solutions; Figure~\ref{fig:can_realize} on page \pageref{fig:can_realize} illustrates these configurations.
\begin{theorem}
	\label{alg-thr}
	Any BARP instance has an optimal solution $S$ with $i\leq 4$ rectangles. Moreover (up to rotating the instance by a multiple of $90^\circ$ and/or reflecting horizontally) the anchor-points $p_1,\dots,p_i$ used by $S$ satisfy one of the following:
	\begin{enumerate}
\vspace*{-2mm}
\itemsep -3pt
		\item $i=1$, and $p_1$ is the leftmost point of $P_\mathcal{T}\cup P_\mathcal{B}$.
		\item $i=2$, and one of the following holds:
		\begin{enumerate}[label = $($\alph*$)$]
\vspace*{-2mm}
\itemsep -1pt
			\item $p_1$ is the bottommost point of $P_\mathcal{L}$ and $p_2$ is the leftmost point of $P_\mathcal{T}\cup P_\mathcal{B}$, or
			\item $p_1$ and $p_2$ are the two points of $P_\mathcal{T}\cup P_\mathcal{B}$ with the closest $x$-coordinates.
		\end{enumerate}
		\item $i=3$,  $p_1\in P_\mathcal{B}$ and $p_2\in P_\mathcal{T}\cup P_\mathcal{B}$ have closest $x$-coordinates with
			$x(p_1)<x(p_2)$, and $p_3$ is the lowest point in $P_\mathcal{L}$.
		\item $i=4$, $p_1\in P_\mathcal{L}$ and $p_3\in P_\mathcal{R}$ have closest $y$-coordinates with $y(p_1)>y(p_3)$, 
			and $p_2\in P_\mathcal{T}$ and $p_4\in P_\mathcal{B}$ have the closest $x$-coordinates with $x(p_4)<x(p_2)$.
	\end{enumerate}
\end{theorem}

\paragraph{Algorithm.} Our algorithm proceeds as follows.  For each of the four rotations, for each of the two reflections, and for each rule 1, 2(a), 2(b), 3, and 4 in Theorem~\ref{alg-thr}, compute the corresponding point set.   Each of these up to 40 point sets defines a cell $H$, and a packing that covers $Q-H$ (see also Lemma~\ref{lem:can_realize_boundary}).
The algorithm returns the one that has the smallest hole $H$.

Having $P_\mathcal{L},P_\mathcal{T},P_\mathcal{R}$, and $P_\mathcal{B}$ sorted along the boundary of $Q$, 
we can also compute sorted lists of $P_\mathcal{L}\cup P_\mathcal{R}$ and $P_\mathcal{T}\cup P_\mathcal{B}$ in linear time.
The closest pair within each or between two of them can be computed in linear time.

The correctness will be proved in Section~\ref{correctness-section}. The proof does not use that $Q$ is a square, only that it
is an axis-aligned rectangle.  We hence have:
\begin{theorem}
	\label{thm:mainResult}
	The boundary anchored rectangle packing problem for $n$ points, given in sorted order on the boundary of a rectangle, can be solved in $O(n)$ time.
\end{theorem}

\section{Correctness of the Algorithm}
\label{correctness-section}

\begin{figure}[tb]
	\centering
	\setlength{\tabcolsep}{0in}
	$\begin{tabular}{cc}
	\multicolumn{1}{m{.5\columnwidth}}{\centering\includegraphics[width=.35\columnwidth]{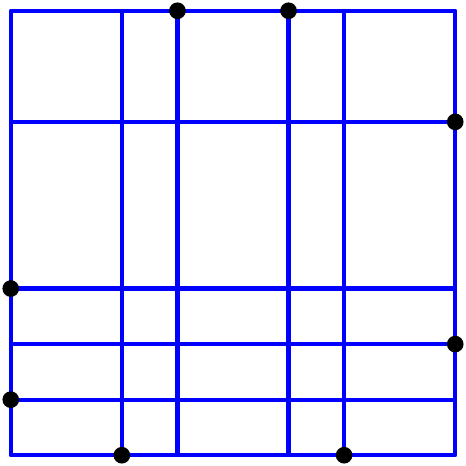}}
	&\multicolumn{1}{m{.5\columnwidth}}{\centering\includegraphics[width=.35\columnwidth]{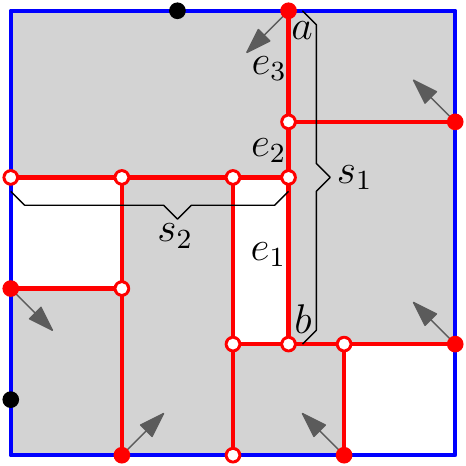}}
	\\
	(a)&(b)
	\end{tabular}$
	\caption{(a) The grid $\grid{P}{Q}$. (b) White regions are holes. Graph $\pgraph{S}$ is in red (thick); filled vertices are points of $P$. The max-segment $s_1$ is introduced while $s_2$ is not.}
	\label{grid-fig}
\end{figure}

We first consider the cases when the square $Q$ can be covered entirely by a packing.

\begin{obs} 
\label{obs:entireQ}
Assume one of the following holds.
	\begin{enumerate}[$($i$)$]
		\vspace*{-2mm}
		\itemsep -1pt
		\item there exists a point $p_1\in P$ on a corner of $Q$, or
		\item there exist two points $p_1,p_2\in P_\mathcal{L}\cup P_\mathcal{R}$ that have the same $y$-coordinates, or
		\item there exist two points $p_1,p_2\in P_\mathcal{T}\cup P_\mathcal{B}$ that have the same $x$-coordinates.
	\end{enumerate}
Then we can cover all of $Q$ with anchored rectangles.
\end{obs}
\begin{proof}
In case (i), one rectangle anchored at $p_1$ can cover all of $Q$.
In case (ii) and (iii), two rectangles anchored at $p_1,p_2$ can
cover all of $Q$.
\end{proof}

Since these conditions are easily tested, 
we assume for most of the remaining section that none of (i)--(iii) holds.    
(We will see that this implies that there must be a hole.)

We need some notation. Throughout this section, let $S$ be an optimal solution for the BARP problem. The term ``rectangle'' now means one of the rectangles used by $S$. Define $\pgraph{S}$ to be the graph whose vertices are the rectangle-corners that are not corners of $Q$, and whose edges are coincident with the rectangle-sides not on the boundary of $Q$; see Figure~\ref{grid-fig}(b). 

We define a {\em max-segment} of $\pgraph{S}$ to be a maximal chain $s$ of collinear edges of $\pgraph{S}$. We say that $s$ is {\em introduced} if at least one endpoint of $s$ belongs to $P$ and is used as anchor-point for some rectangle of $S$. For example, in Figure~\ref{grid-fig}(b), $s_1$ is introduced but $s_2$ is not. Every edge $e$ belongs to exactly one max-segment $s_e$; we say that $e$ is {\em introduced} if $s_e$ is. In Figure~\ref{grid-fig}(b), $e_1$, $e_2$ and $e_3$ are introduced. We already know~\cite{BalasDT16} that all boundaries of rectangles can be assumed to lie on the grid $\grid{P}{Q}$, but we need to strengthen this and prove the following: 
\begin{lemma}
	\label{grid-lemma}
	There exists an optimal solution $S$ such that all max-segments of $S$ are introduced. 
\end{lemma}

\begin{proof}
Let $S$ be an optimal solution that, among all optimal solutions, minimizes the number of max-segments.
Assume for contradiction that there exists a max-segment $s$ that is not introduced. 
After rotation we may assume that $s$ is horizontal.   Let $V$ be the vertical slab defined by the 
two vertical lines through the endpoints of $s$; see Figure~\ref{fig:sidesOfH}. 

	\begin{figure}[tb]
		\centering
		\includegraphics[width=.60\textwidth]{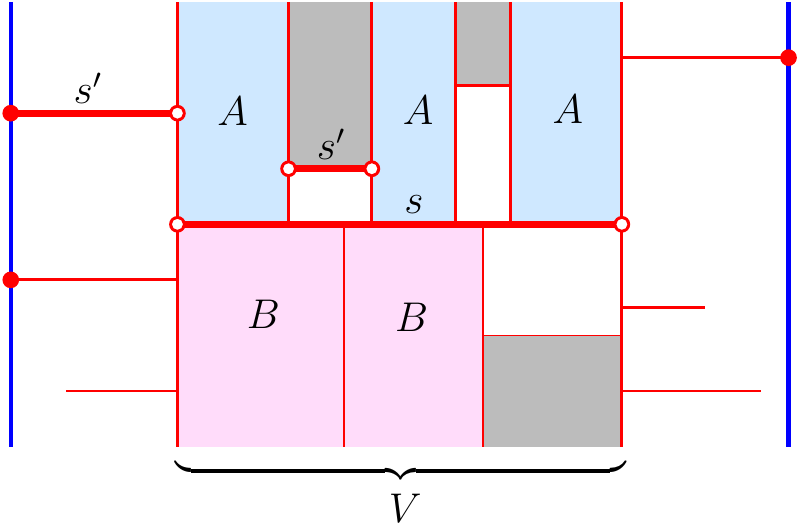}
		\caption{Illustration of the proof of Lemma~\ref{grid-lemma}.}
		\label{fig:sidesOfH}
	\end{figure}
	
Consider moving $s$ upward in parallel, i.e., shortening the rectangles $A$ with their bottom sides on $s$ and lengthening the rectangles $B$ with their top sides on $s$. Observe first that these rectangles indeed can be shortened/lengthened, because none of them can be anchored at a point on $s$: the only points of $s$ that are possibly in $P$ are its ends, but neither of them anchors a rectangle since $s$ is not introduced. If this move of $s$ increases the coverage, then $S$ was not optimal, a contradiction. If this decreases the coverage, then moving downward in parallel would increase the coverage, a contradiction. So the covered area must remain the same during the move. Shift $s$ up until it hits either the boundary of $Q$ or intersects some other horizontal max-segment of $\pgraph{S}$. If $s$ hits the boundary of $Q$, then $s$ disappears and will be deleted from $\pgraph{S}$. If $s$ intersects some other horizontal max-segment $s'$ of $\pgraph{S}$ (which may be inside $V$ or only share an endpoint with the translated $s$) then the two max-segments merge into one. Either way we decrease the number of max-segments, which contradicts the choice of $S$ and proves the lemma.
\end{proof}

From now on, without further mentioning, we assume that $S$ is an optimal solution
where all max-segments are introduced. We also assume that, among all such optimal solutions, $S$ minimizes the number of rectangles.
\begin{lemma}
	\label{degree-lemma}
	Every internal vertex of $\pgraph{S}$ has degree three or four.
\end{lemma}
\begin{proof}
Every internal vertex of $\pgraph{S}$ resides on the corner(s) of axis-aligned rectangle(s), and so has degree at least 2 and at most 4. Assume for contradiction that there is a vertex $b$ of $\pgraph{S}$ that has degree exactly 2, and
let $a$ and $c$ be its neighbours. After possible rotation, we may assume that 
$a$ lies to the left of $b$, and $c$ lies above $b$, as depicted in Figure~\ref{degree-fig}. Thus, $b$ is the bottom-right corner of some rectangle $r_1$, and no other rectangle has $b$ on its boundary. This implies that the region to the right of $bc$ and below $ab$ belongs to some hole $H$. So rectangle $r_1$ is anchored either on the left or the top side of $Q$; after a possible diagonal flip we assume that it is anchored on the left.
	
	\begin{figure}[tb]
		\centering
		\setlength{\tabcolsep}{0in}
		$\begin{tabular}{cc}
		\multicolumn{1}{m{.5\columnwidth}}{\centering\includegraphics[width=.47\columnwidth,page=1]{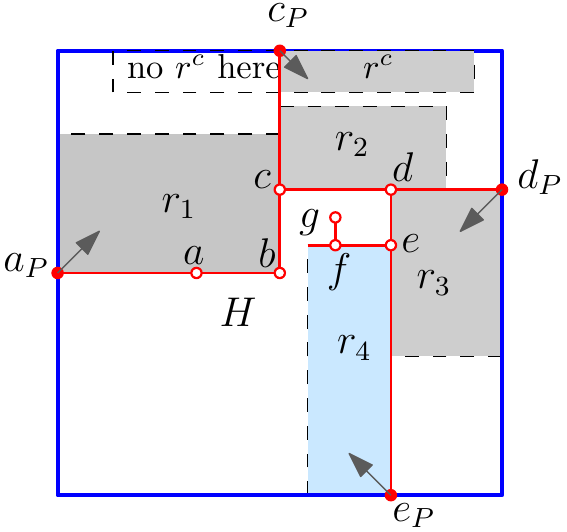}}
		&\multicolumn{1}{m{.5\columnwidth}}{\centering\includegraphics[width=.47\columnwidth,page=2]{deg3-1.pdf}}
		\\
		(a)&(b)
		\end{tabular}$
		\caption{Illustration of the proof of Lemma~\ref{degree-lemma}.}
		\label{degree-fig}
	\end{figure} 

Define $a_P$ and $c_P$ to be the points of $P$ that introduced $ab$ and $cb$, respectively;
we know that these must be on $P_\mathcal{L}$ respectively $P_\mathcal{T}$ since $b$ has degree 2.
By definition of ``introduced'' some rectangle $r^c$ is anchored at $c_P$.  We claim that $r^c$ cannot have $c_P$ as its top-right corner. Assume for contradiction that it has. Then we can expand $r^c$ (if needed) to cover the entire rectangle spanned by $a_P$ and $c_P$; this can only increase the coverage. In particular, the expanded $r^c$ covers all of $r_1$. We know that $r_1\neq r^c$ since $r_1$ was anchored on the left side of $Q$.  This contradicts that $S$ has the minimum number of rectangles, so $r^c$ has $c_P$ as its top-left corner.

If the right side $\rs{r_1}$ of $r_1$ is a sub-segment of $bc$, then we can stretch $r_1$ to the right to increase the coverage of $S$, contradicting optimality. So $\rs{r_1}$ must be a strict super-segment of $bc$, which in particular implies that $c$ is interior and has no leftward edge.  Since $c$ is a vertex, it must have a rightward edge; let $d$ be the vertex of $H$ to the right of $c$. Let $r_2$ be the rectangle whose bottom-left corner is $c$; this exists since edge $cd$ is the boundary of some rectangle(s), but the area below $cd$ belongs to hole $H$. Rectangle $r_2$ cannot be anchored on the right, otherwise we could expand $r^c$ to cover all of $r_2$ and reduce the number of rectangles, a contradiction. So $r_2$ is anchored on the top, which implies that $r_2=r^c$, else they would overlap.

If the bottom side $\bs{r_2}$ of $r_2$ is a sub-segment of $cd$, then we can stretch $r_2$ down to increase the coverage of $S$. So $\bs{r_2}$ is a strict super-segment of $cd$, which implies that $d$ is interior. We iterate this process three times as follows. (i) Let $e$ be the vertex of $H$ that is below $d$, and let $r_3$ be the rectangle whose top-left corner is $d$. Argue as before that $r_3$ is anchored at the right endpoint $d_P$ of the max-segment through $cd$, therefore the left side $\ls{r_3}$ is a strict super-segment of $de$ and $e$ is interior. (ii) Let $f$ be the vertex of $H$ that is to the left of $e$, and let $r_4$ be the rectangle whose top-right corner is $e$. Argue as before that $r_4$ is anchored at the bottom endpoint $e_P$ of the max-segment through $de$, therefore the top side $\ts{r_4}$ is a strict super-segment of $e\!f$ and $f$ is interior. (iii)~Finally, let $g$ be the vertex of $H$ that is above $f$ (possibly $g=a$). Now observe that the max-segment through $f\!g$ cannot reach the boundary of $Q$ without intersecting $r_4,r_1$ or $r_2$. Therefore, $f\!g$ is not introduced---a contradiction.
\end{proof}

We assumed that neither case (ii) nor (iii) of Observation~\ref{obs:entireQ} holds, which means that any grid-line of the grid $\grid{P}{Q}$ has exactly one end in $P$. So, we can direct the edges of the grid (and with it the edges of $\pgraph{S}$) from the end in $P$ to the end not in $P$. See also Figure~\ref{fig:can_realize}. Define a {\em guillotine cut} to be a max-segment of $\pgraph{S}$ for which both endpoints are on the boundary of $Q$.
\begin{lemma}
\label{max-segment-lemma}
\label{lem:guillotine}
\label{lem:one_hole_interior}
If there is no guillotine cut, then $S$ has a hole $H$. Furthermore, $H$ is a rectangle, $H$ is not incident to the boundary of $Q$, and the boundary of $H$ is a directed cycle of $\pgraph{S}$.
\end{lemma}
\begin{proof}
We claim that no vertex $w$ of $\pgraph{S}$ on the boundary of $Q$ is a sink. For if the unique edge incident to $w$ were directed from some vertex $v$ to $w$, then by Lemma~\ref{grid-lemma} and the way we directed the edges of $G(S)$, the point $p$ that introduced $vw$ would be on the opposite side and hence the max-segment $pw$ would be a guillotine cut. Likewise no interior vertex $w$ can be a sink,  because $\deg(w)\geq 3$ by the previous lemma, which implies that two incident edge of $w$ have the same orientation (horizontal or vertical). One of them then becomes outgoing at $w$ since we direct edges along grid-lines. So $\pgraph{S}$ has no sink, which implies that it has a directed cycle $C$.
The region enclosed by $C$ has no point on
the boundary, so no rectangle anchored on the boundary can cover parts of it without intersecting $C$.  
So the interior region of $C$ is a hole $H$ not incident to the boundary.
We know that $H$ is a rectangle since it has no vertex of degree 2
by the previous lemma, hence in particular no reflex vertex.
\end{proof}

We use the previous lemma in the proof of our following stronger claim.
\begin{lemma}
\label{one-hole-lemma}
\label{lem:one_hole_all}
If $S$ has holes, then it has a hole $H$ that is a rectangle. Furthermore, every interior corner of $H$ has an incoming edge that lies on $H$.
\end{lemma}
\begin{proof}

We prove by induction on the cardinality of $S$. If $S$ is empty then $H=Q$ and we are done. Assume that $S$ is not empty. If there is no guillotine cut, then $H$ is a rectangle by Lemma~\ref{lem:guillotine} which is interior and whose boundary is a directed cycle; this satisfies all claims. Assume now that there is a guillotine-cut $aa'$, say it is horizontal. Since case (ii) of Observation~\ref{obs:entireQ} does not hold, not both $a$ and $a'$ can belong to $P$, say $a'\not\in P$. Segment $aa'$ divides $Q$ into two rectangles $Q_1$ and $Q_2$ with $Q_1$ above $Q_2$; see Figure~\ref{fig:one_hole_all}(a). There is a rectangle $r_1$ that is anchored at $a$; up to a vertical flip we may assume that $r_1$ is inside $Q_1$. Observe that $r_1$ must cover all of $Q_1$, else we could find a solution with more coverage or fewer rectangles. Thus $S':=S\setminus \{r_1\}$ is an anchored-rectangle packing for $Q_2$ with anchor-points in $P\setminus \{a\}$. Notice that $S'$ must be optimal for $Q_2$, else we could get a better packing for $Q$ by adding $r_1$ to it.  It cannot cover all of $Q_2$ since $S$ had holes. Therefore, $S'$ has a hole $H$ and by induction hypothesis $H$ is a rectangle; observe that $H$ is also a rectangular hole in $S$. This finishes our proof of the first claim.

\begin{figure}[t]
\hspace*{\fill}
\begin{subfigure}{0.31\linewidth}
\includegraphics[width=0.99\linewidth,page=1]{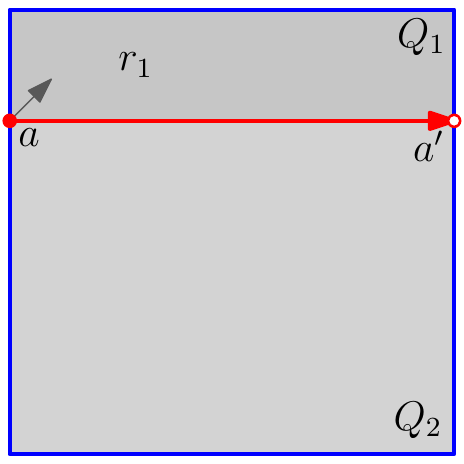}
\caption{}
\end{subfigure}
\hspace*{\fill}
\begin{subfigure}{0.31\linewidth}
\includegraphics[width=0.99\linewidth,page=2]{one_hole_all.pdf}
\caption{}
\end{subfigure}
\hspace*{\fill}
\begin{subfigure}{0.31\linewidth}
\includegraphics[width=0.99\linewidth,page=4]{one_hole_all.pdf}
\caption{}
\end{subfigure}
\caption{With a guillotine cut, a hole can be found in $Q_2$ recursively.}
\label{fig:one_hole_all}
\end{figure}

Now, consider a corner $c$ of $H$ that is interior to $Q$. If $c$ is interior to $Q_2$, then by induction hypothesis it has an incoming edge that lies on $H$. Assume that $c$ is not interior to $Q_2$ and thus it lies on $aa'$ and $c\neq a,a'$. Consider an interior vertical edge $e$ of $H$; see Figures~\ref{fig:one_hole_all}(b) and \ref{fig:one_hole_all}(c). Observe that $e$ is directed upwards because otherwise the max-segment $s_e$ containing it would have intersect $r_1$. It follows that all interior vertical edges of $H$ are directed upwards. Therefore the vertical edge of $H$ that is incident to $c$ is also upward and hence incoming to $c$ as desired.
\end{proof}

Hence, hole $H$ must satisfy this \emph{hole-condition} on the edge directions (at least for \emph{some} optimal solution $S$); that is, every interior corner of $H$ has an incoming edge that lies on $H$. It turns out that this condition is also sufficient.
\begin{lemma}
\label{lem:can_realize_interior}
\label{lem:can_realize_boundary}
\label{lem:can_realize}
Let $H$ be a rectangle whose sides lie on $\grid{P}{Q}$. If $H$ satisfies the hole-condition, then there exists a packing that covers $Q\setminus H$.
\end{lemma}
\begin{proof}
Let $p_1,\dots,p_i$ (for some $i\leq 4$) be the points of $P$ that defined
the grid lines on which the sides of $H$ reside.  We distinguish cases (1)--(4)
depending on how many sides of $H$ are interior, where case (2) splits further
into cases (2a) and (2b) depending on whether the sides are adjacent or parallel.
After possible rotation, the hole is situated as shown in Figure~\ref{fig:can_realize}.
Every interior corner of $H$ has an incoming edge that is on $H$, which (up to
reflection) forces the location of some of $p_1,\dots,p_i$ as indicated in
the figure. In all cases, one verifies that $i$ rectangles anchored at $p_1,\dots,p_i$ suffice to cover $Q\setminus H$.
\end{proof}

\begin{figure}[t]
\includegraphics[width=0.30\linewidth,page=5]{can_realize.pdf}%
\hspace*{\fill}
\includegraphics[width=0.30\linewidth,page=3]{can_realize.pdf}%
\hspace*{\fill}
\includegraphics[width=0.30\linewidth,page=4]{can_realize.pdf}%
\hspace*{\fill}
\newline
\hspace*{\fill}
\includegraphics[width=0.30\linewidth,page=2]{can_realize.pdf}
\hspace*{\fill}
\includegraphics[width=0.30\linewidth,page=1]{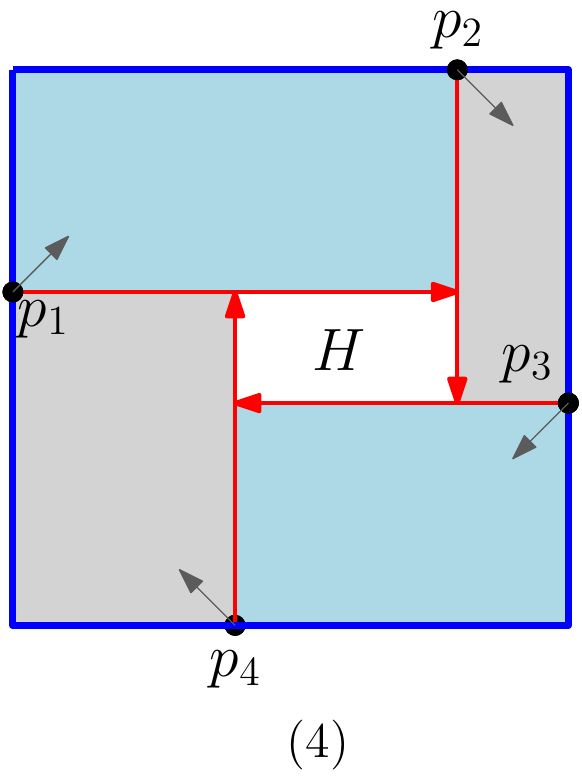}
\hspace*{\fill}
\caption{Any rectangle whose boundary is directed suitably can be realized as hole.}
\label{fig:can_realize}
\end{figure}

We are now ready to prove Theorem~\ref{ins:standardFormOtimal}. To this end, we first show the following:
\begin{lemma}
\label{lem:one_hole_cell}
If $S$ has holes, then it has exactly one hole $H$, and $H$ is a cell of~$\grid{P}{Q}$.
\end{lemma}
\begin{proof}
Lemma~\ref{lem:one_hole_all} shows that we may assume $H$ to be a rectangle that satisfies the hole-condition. By Lemma~\ref{lem:can_realize_interior},  we can cover $Q\setminus H$ with anchored rectangles, which by maximality of $S$ means that $H$ is unique.

If $H$ is not a cell, then it is bisected by some grid-line $\ell$ into two pieces $H_1$ and $H_2$. If some $H'\in \{H_1,H_2\}$ satisfies the hole-condition (i.e., all interior corners have incoming edges on $H'$), then we can create a packing that covers $Q\setminus H' \supset Q\setminus H$, which contradicts the minimality of $S$. In fact, by inspecting the possible configurations of $H$ in cases (1), (2a), (2b), (3), and (4), as well as possible placements of the ``undecided'' anchor-points and the orientation/direction of $\ell$ (see Figure~\ref{fig:bisect}, which shows all but one case), we observe that $H'$ satisfies the hole-condition as we can cover $Q\setminus H'$ in each of these cases. So, there is a contradiction in all cases, and $H$ must be one cell.
\end{proof}

\begin{figure}[t]
\includegraphics[width=0.2\linewidth,page=1]{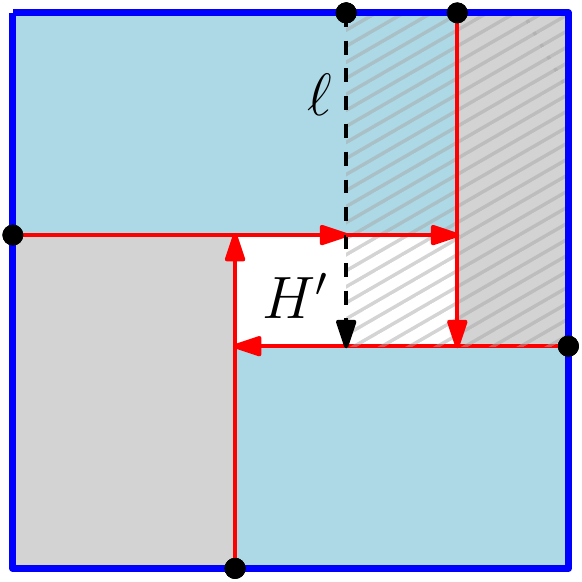}%
\hspace*{\fill}
\includegraphics[width=0.2\linewidth,page=2]{bisect.pdf}%
\hspace*{\fill}
\includegraphics[width=0.2\linewidth,page=3]{bisect.pdf}%
\hspace*{\fill}
\includegraphics[width=0.2\linewidth,page=4]{bisect.pdf}%
\newline
\includegraphics[width=0.2\linewidth,page=5]{bisect.pdf}%
\hspace*{\fill}
\includegraphics[width=0.2\linewidth,page=6]{bisect.pdf}%
\hspace*{\fill}
\includegraphics[width=0.2\linewidth,page=7]{bisect.pdf}%
\hspace*{\fill}
\includegraphics[width=0.2\linewidth,page=8]{bisect.pdf}%
\newline
\includegraphics[width=0.2\linewidth,page=9]{bisect.pdf}%
\hspace*{\fill}
\includegraphics[width=0.2\linewidth,page=10]{bisect.pdf}%
\hspace*{\fill}
\includegraphics[width=0.2\linewidth,page=11]{bisect.pdf}%
\hspace*{\fill}
\includegraphics[width=0.2\linewidth,page=12]{bisect.pdf}%
\newline
\includegraphics[width=0.2\linewidth,page=13]{bisect.pdf}%
\hspace*{\fill}
\includegraphics[width=0.2\linewidth,page=14]{bisect.pdf}%
\hspace*{\fill}
\includegraphics[width=0.2\linewidth,page=15]{bisect.pdf}%
\hspace*{\fill}
\includegraphics[width=0.2\linewidth,page=16]{bisect.pdf}%
\newline
\caption{If a hole that satisfies the hole-condition is bisected by a line $\ell$, then this gives rise to a smaller hole $H'$ that satisfies the hole-condition.}
\label{fig:bisect}
\end{figure}

By Lemma~\ref{lem:one_hole_cell}, we have characterized solutions that have holes.  It remains to characterize solutions that do not have holes; i.e., to show that conditions (i)--(iii) of Observation~\ref{obs:entireQ} are necessary.
\begin{lemma}
\label{entire-Q-lemma}
\label{lem:entireQ}
  	If $Q$ can be covered with anchored rectangles, then one of (i), (ii) or (iii) holds.
\end{lemma}
\begin{proof}
Let $S$ be a packing that covers all of $Q$. If $\pgraph{S}$ has no edge, then all of $Q$ must be covered by one rectangle, which hence must be anchored at a corner of $Q$ and (i) holds. So assume that $\pgraph{S}$ has edges. By Lemma~\ref{max-segment-lemma}, since $S$ has no hole there must be a guillotine-cut $aa'$, say it is horizontal. If both $a$ and $a'$ are in $P$ then (ii) holds and we are done, so assume $a\in P$ and $a'\notin P$.

Recall the rectangles $Q_1, Q_2$ and $r_1$ from the proof of Lemma~\ref{lem:one_hole_all} and observe that $S':=S\setminus \{r_1\}$  covers all of $Q_2$ using anchor-points in $P':=P\setminus \{a\}$. Apply induction to $S',P',Q_2$. If (i) holds for them, then $P'$ has a point on a corner of $Q_2$, which by $a,a'\notin P'$ is also a corner of $Q$ and we are done. If (ii) holds for them, then two points in $P'\subset P$ have the same $y$-coordinate and we are done. Finally (iii) cannot hold for $S',P',Q_2$ because the top side of $Q_2$ has no point of $P'$ on it since $a'\not\in P$.
\end{proof}

We are finally ready to prove Theorem~\ref{alg-thr}. Let $S$ be the optimal solution with the minimum number of rectangles. If $S$ covers all of $Q$, then by Lemma~\ref{lem:entireQ} one of (i)--(iii) holds. If (i) holds, then the corner in $P$ will be chosen under rule (1). (In these and all other cases, ``chosen'' means ``after a suitable rotation and/or reflection''.) If (ii) or (iii) holds then the two points with the coinciding coordinate will be chosen under rule (2b).

If $S$ has holes, then by Lemma~\ref{lem:one_hole_all}
its unique hole $H$ is a cell such that
all interior corners of $H$ have incoming edges on $H$.
Let $p_1,\dots,p_i$ be the points that introduce interior sides of $H$.
We know that $H$ has one of the types shown in Figure~\ref{fig:can_realize},
and $p_1,\dots,p_i$ hence will be considered
under the corresponding rule.  Moreover, all point sets that fit
the type can be realized by Lemma~\ref{lem:can_realize}. So $H$
must be the one that minimizes the area, which corresponds to the
points minimizing the $x$-distance resp.~$y$-distance.  So 
one of rules 1, 2a, 2b, 3 or 4 applies to the points $p_1,\dots,p_i$ 
and Theorem~\ref{alg-thr} holds.

\section{Boundary-Anchored Squares}
\label{sec:squares}
Recall that $Q$ is an axis-aligned square in the plane and $P$ is a set of points on the boundary of $Q$. In the \emph{boundary anchored square packing} (BASP) problem we want to find a set of disjoint axis-aligned squares in $Q$ that are anchored at points of $P$ and maximize the total area. For this problem we are unable to find a grid\textemdash as we did for rectangle packing\textemdash that discretizes the problem such that the sides of every square in an optimal solution lie on that grid. It might be tempting to obtain a grid as follows. For every point $p$ on the bottom side of $Q$ we add the following lines to the grid (see Figure~\ref{square-grid}(a)): 
\begin{enumerate}[(\arabic*)]
	\item one vertical line through $p$,
	\item one horizontal line through the top side of the largest square in $Q$ that has $p$ on its bottom-left corner, and one for a similar square that has $p$ on its bottom-right corner, and
	\item for every other point $q$ on the bottom side of $Q$, we add one horizontal line through the top side of the square that has the segment $pq$ as its bottom side.
\end{enumerate}

We add similar lines for points that are on the left side, the right side, and the top side of $Q$. Let $\grid{P}{Q}$ be the resulting grid. We construct a set of points for which no optimal solution of the BASP is introduced by $\grid{P}{Q}$. Figure~\ref{square-grid}(b) shows a set of six points with an optimal solution associated for them. A point $p$ lies on the bottom side of $Q$ and at distance $\delta$ from the bottom-left corner of $Q$, for a small $\delta>0$. Five points $u,v,w,x,y$ arranged on the top side of $Q$ from left to right such that $w$ is the mid-point of the top side of $Q$, $|vw|=|wx|=1.5\delta$, and $|uv|=|xy|=\epsilon$, for a small $\epsilon$ that is much less than $\delta$. Any optimal solution for this setting contains the largest square in $Q$ that has $p$ on its bottom-left corner. Also any optimal solution contains the two squares that are anchored at $u$ and $y$ as depicted in Figure~\ref{square-grid}(b). The solution shown in Figure~\ref{square-grid}(b) is optimal. Any optimal solution contains two squares of side-length $\delta$ and one square of side-length $\delta/2$ that are anchored at $v,w,x$. The square of side-length $\delta/2$ is not defined by $\grid{P}{Q}$, no matter on which of $v,w,x$ it is anchored.

\begin{figure}[t]
	\centering
	\setlength{\tabcolsep}{0in}
	$\begin{tabular}{cc}
	\multicolumn{1}{m{.45\columnwidth}}{\centering\includegraphics[width=.3\columnwidth]{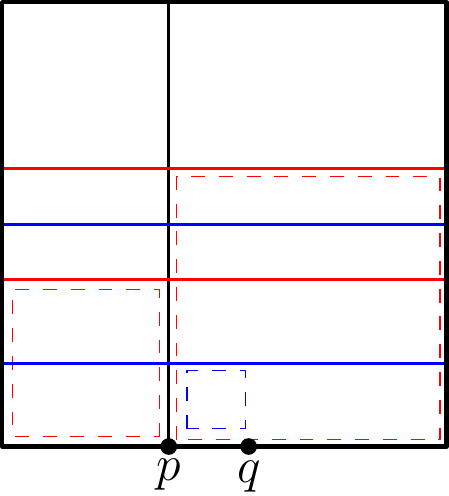}}
	&\multicolumn{1}{m{.55\columnwidth}}{\centering\includegraphics[width=.33\columnwidth]{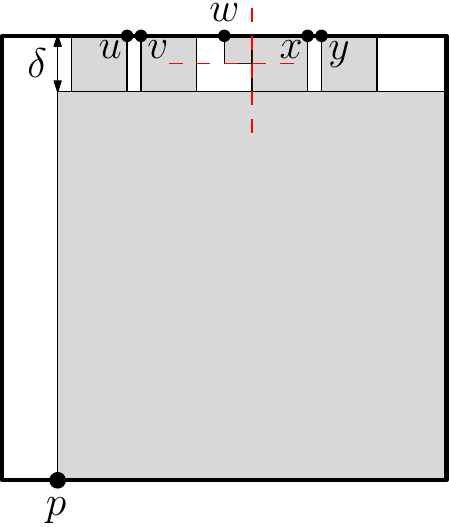}}
	\\
	(a)&(b)
	\end{tabular}$
	\caption{(a) The grid lines for every point $p$. (b) An optimal solution in which the square anchored at $w$ is not introduced by $\grid{P}{Q}$.}
	\label{square-grid}
\end{figure}

In the rest of this section we consider two special cases where the points of $P$ lie only on one side of $Q$, or on two opposite sides of $Q$. Later we will see that the two opposite-side case can be reduced to some instances of the one-side case. 

\subsection{Points on one side}
In this section we consider a version of the BASP problem where the points of $P$ lie only on one side of $Q$. We consider a more general version where $Q$ is rectangle and the points of $P$ lie on a longer side of $Q$. To avoid confusion in our notation, we use $R$ to represent such a $Q$. Let $w$ and $h$ denote the width and height of $R$, respectively. We assume that the longer side of $R$ is parallel to the $x$-axis and points of $P$ lie on the bottom side of $R$; see Figure~\ref{rectangle-fig}. We introduce a grid $\ugrid{}{}$ such that any optimal solution for this problem is defined by $\ugrid{}{}$. This grid contains the following lines:
\begin{enumerate}[(\arabic*)]
	\item a vertical line through $p$,
	\item a horizontal line through the top side of the largest square in $R$ that has $p$ on its bottom left-corner, and one for a similar square that has $p$ on its bottom-right corner, 
	\item for every other point $q$, that is at distance at most $h$ from $p$, we add one horizontal line through the top side of the square that has the segment $pq$ as its bottom side, and
	\item one vertical line through the right side of the largest square in $R$ that has $p$ on its bottom-left corner, and one for a similar square that has $p$ on its bottom-right corner.
\end{enumerate}

Based on the construction of $\ugrid{}{}$, we define a set $\mathcal{S}$ of squares that are obtained as follows. For every point $p\in P$ we add to $\mathcal{S}$ three types of squares (see Figure~\ref{rectangle-fig}(a)):
\begin{enumerate} [(a)]
	\item[T1] the largest square in $R$ that has $p$ on its bottom-left corner, and the largest square in $R$ that has $p$ on its bottom-right corner,
	\item[T2] for every other point $q$, that is within distance $h$ from $p$, we add a square that has the segment $pq$ as its bottom side, and
	\item[T3] for every other point $q$ to the right (resp. left) of $p$, which is within distance $2h$ from $p$, we add a square of side length $|pq|-h$ that has $p$ on its bottom-left corner (resp. bottom-right corner).
\end{enumerate}
The set $\mathcal{S}$ contains $O(n^2)$ squares and all of them are introduced by $\ugrid{}{}$. We say that a square is {\em introduced} by a grid if at least three of its sides lie on the grid.  The following lemma enables us to discretize the problem.
\begin{lemma}
	\label{inclusion-lemma}
Consider an optimal solution for the variant of the BASP problem where the points lie only on the bottom side of the input rectangle. Then, all squares of the solution belong to $\mathcal{S}$.
\end{lemma} 
\begin{proof}
	Our proof is by contradiction. Consider an optimal solution $S$ for this problem and assume that it contains a square $s$ that does not belong to $\mathcal{S}$. Without loss of generality we assume that $s$ has a point $p$ on its bottom-left corner. Since $s$ is not of type (T1), the top side of $s$ does not lie on the top side of $R$. Also, the right side of $s$ does not lie on the right side of $R$. If the right side of $s$ does not touch any other square in $S$, then we can enlarge $s$ and increase the total area of $S$ which contradicts its optimality. Let $r$ be the square that touches the right side of $s$. Let $q$ be the point that $r$ is anchored on. Since $s$ is not of type (T2), $q$ is the bottom-right corner of $r$. Moreover, since $s$ is not of type (T3), $r$ is not a largest square that is anchored at $q$.

	So, we have two touching squares $s$ and $r$ and none of them are maximum squares. See Figure~\ref{rectangle-fig}(b). Without loss of generality assume that $s$ is not smaller than $r$. By concurrently enlarging $s$ and shrinking $r$ by a small amount, the gain in the area of $s$ would be larger than the loss in the area of $r$. This will increase the total area of $S$ which contradicts its optimality.
\end{proof}

\begin{figure}[t]
	\centering
	\setlength{\tabcolsep}{0in}
	$\begin{tabular}{cc}
	\multicolumn{1}{m{.54\columnwidth}}{\centering\includegraphics[width=.35\columnwidth]{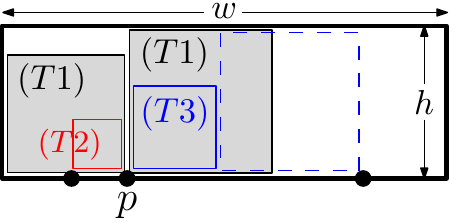}}
	&\multicolumn{1}{m{.46\columnwidth}}{\centering\includegraphics[width=.338\columnwidth]{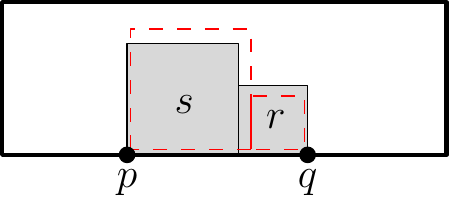}}
	\\
	(a)&(b)
	\end{tabular}$
	\caption{(a) The construction of $\mathcal{S}$ (b) Illustration of the proof of Lemma~\ref{inclusion-lemma}.}
	\label{rectangle-fig}
\end{figure}

As a consequence of Lemma~\ref{inclusion-lemma}, to solve the BASP problem, it suffices to find a subset of non-overlapping squares in $\mathcal{S}$ with maximum area. For every square $s\in \mathcal{S}$, we introduce a closed interval $I_s$ with the bottom side of $s$. We set the weight of $I_s$ to be the area of $s$. Let $\mathcal{I}$ be the set of these intervals. Any maximum-weight independent set of intervals in $\mathcal{I}$ corresponds to a set of non-overlapping squares in $\mathcal{S}$ with maximum area. A maximum-weight independent set of $m$ intervals, given in sorted order of their left endpoints, can be computed in $O(m)$ time~\cite{Hsiao1992}. The set $\mathcal{S}$ contains $O(n^2)$ squares and can be computed within the same time bound. Consequently, $\mathcal{I}$ can be computed in $O(n^2)$ time. Having the points of $P$ sorted from left to right, the sorted order of the intervals in $\mathcal{I}$ can be obtained within the same time bound. Thus, the total running time of our algorithm is $O(n^2)$.

\subsection{Points on two opposite sides}
In this section we study a version of the BASP problem where the points of $P$ lie on two opposite sides of square $Q$. We show how to reduce an instance of this problem into $O(n^2)$ instances of the one-sided version. Since the one-sided version can be solved in $O(n^2)$ time, this reduction implies an $O(n^4)$-time solution for the two-sided version. We refer to a square that is anchored at a top point (resp. bottom point) by a {\em top square} (resp. a {\em bottom square}).   
\begin{lemma}
\label{separating-line-lemma}
For any optimal solution for the BASP problem, where the input points lie only on the top and bottom sides of the input square, there exists no horizontal line $\ell$ that intersects both a top square and a bottom square in its interior.
\end{lemma}
\begin{proof}
Suppose for a contradiction that such line $\ell$ exists, say it intersects a top square $s$ and a bottom square $r$.
Since $\ell$ crosses both $s$ and $r$, the height of $s$ plus the height of $r$ is larger than $h$ (the height of the boundary square). This implies that their total width is also larger than $h$. Since $s$ and $r$ are non-overlapping, there is a vertical line which separates $s$ from $r$. These two facts imply that the width of the boundary square is larger than $h$, which is a contradiction.
\end{proof}

By Lemma~\ref{separating-line-lemma}, for every optimal solution there exists a horizontal line that separates its top squares from its bottom squares; refer to such a line as a {\em separating line}. We introduce a set $\mathcal{L}$ of $O(n^2)$ horizontal lines and claim that for every optimal solution of the BASP problem, there exists a separating line that belongs to $\mathcal{L}$. Assume that $Q$ is the unit square, and its bottom-left corner is the origin. For a point $p$, let $p_x$ denotes its $x$-coordinate. First, we add to $\mathcal{L}$ the horizontal line $y=1/2$. Then, for every point $p$ on the bottom side of $Q$ we add the following lines to $\mathcal{L}$ (see Figure~\ref{discrete-fig}(a)): 
\begin{enumerate}[(\arabic*)]
	\item $y=p_x$; this line represents the top side of the largest square that has $p$ on its bottom-right corner.
	\item $y=1-p_x$; this line represents the top side of the largest square that has $p$ on its bottom-left corner.
	\item for every point $q\neq p$ on the bottom side of $Q$, we add $y=|p_x-q_x|$; this line represents the top side of the square that is anchored at $p$ and has another corner at $q$.
	\item for every point $q\neq p$ on the bottom side of $Q$, we add $y=|p_x-q_x|/2$; this line represents the the top side of the square that is anchored at $p$ and has another corner at the mid-point of the segment $pq$.
\end{enumerate}

Then, for every point $p$ on the top side of $Q$, we add to $\mathcal{L}$ the lines analogous to items (1)--(4).

\begin{lemma}
\label{discrete-lemma}
For any optimal solution of the BASP problem, there exists a separating line that belongs to $\mathcal{L}$. 
\end{lemma} 
\begin{proof}
	Consider an optimal solution $S$ for this problem. Let $s$ be the largest square in $S$ and, without loss of generality, assume that $s$ is a bottom-square and that $p$ is the bottom-left corner of $s$; see Figure~\ref{discrete-fig}(b). By Lemma~\ref{separating-line-lemma}, there exists a separating line for $S$. Let $\ell$ be the line that touches the top side of $s$; observe that $\ell$ is separating since it intersects no top square by Lemma~\ref{separating-line-lemma}. If $\ell$ is below the line $y=1/2$, then due to maximality of $s$, $y=1/2$ is also a separating line for $S$ and belongs to $\mathcal{L}$. Assume that $\ell$ is above $y=1/2$.

\begin{figure}[t]
	\centering
	\setlength{\tabcolsep}{0in}
	$\begin{tabular}{cc}
	\multicolumn{1}{m{.45\columnwidth}}{\centering\includegraphics[width=.3\columnwidth]{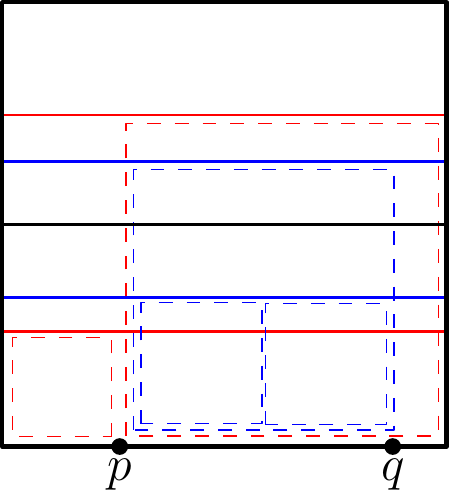}}
	&\multicolumn{1}{m{.55\columnwidth}}{\centering\includegraphics[width=.37\columnwidth]{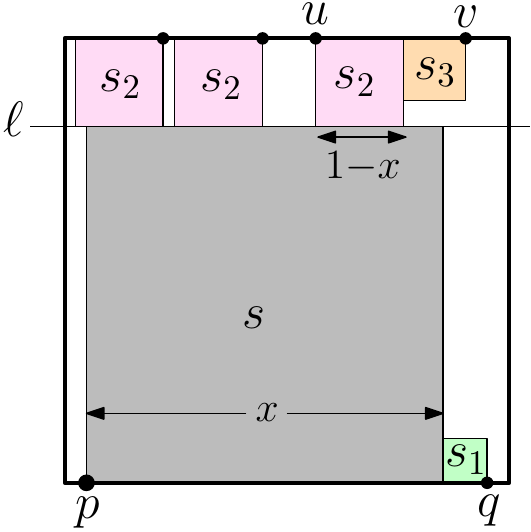}}
	\\
	(a)&(b)
	\end{tabular}$
	\caption{(a) The lines that are added to $\mathcal{L}$ for $p$. (b) Illustration of the proof of Lemma~\ref{discrete-lemma}.}
	\label{discrete-fig}
\end{figure}

The rest of our proof is by contradiction. By a similar reasoning as in the proof of Lemma~\ref{separating-line-lemma}, we argue that $s$ is the only bottom square that touches $\ell$. However, there might be arbitrarily many top squares that touch $\ell$. Let $a$ denote the $y$-coordinate of $\ell$. We continuously move $\ell$ up and down within the vertical range $[a-\epsilon,a+\epsilon]$, for a very small amount $\epsilon$. Then, the $y$-coordinate of $\ell$ is $x$, where $x\in[a-\epsilon,a+\epsilon]$. While moving $\ell$ in this range, we change (enlarge or shrink) some squares of $S$ as follows and keep track of their area (see Figure~\ref{discrete-fig}(b)): 

	\begin{itemize}
		\item We change $s$ in such a way that its top side always lies on $\ell$. Thus, the are of $s$ would be $x^2$.
		\item Observe that the right side of $s$ does not touch the boundary square because otherwise $\ell$ would have been added to $\mathcal{L}$ by item (2). There can be only one square in $S$ that touches the right side of $s$. If such a square exists, then let $s_1$ denote that square, and assume that it is anchored at a point $q$; see Figure~\ref{discrete-fig}(b). The point $q$ is on the bottom-right corner of $s_1$ because otherwise $\ell$ would have been added to $\mathcal{L}$ by item (3). We change $s_1$ in such a way that its left side always touches the right side of $s$. Thus, the area of $s_1$ is $(|pq|-x)^2$
		\item Let $S_2$ be the set of all top squares that touch $\ell$. We change these squares in such a way that their bottom sides touch $\ell$. The area of every such square is $(1-x)^2$.
		\item We construct a set $S_3$ of top squares as follows. Consider every square $s_2\in S_2$ and let $s_2$ be anchored at a point $u$. If there is a top square $s_3$ in $S$ that touches $s_2$ from the side that does not contain $u$, then we add $s_3$ to $S_3$. Let $s_3$ be anchored at $v$ as depicted in Figure~\ref{discrete-fig}(b). The point $v$ is not on the boundary of $s_2$ because otherwise $\ell$ would have been added to $\mathcal{L}$ by item (3). Also, the square $s_3$ does have the same size as $s_2$ because otherwise $\ell$ would have been added to $\mathcal{L}$ by item (4); in fact $s_3$ is smaller than $s_2$. We change $s_3$ in such a way that it always touches $s_2$. Thus, by moving $\ell$ in the above range, the area of $s_3$ will be $(|uv|-(1-x))^2$.
	\end{itemize}

Let $S'$ be the set of the above squares; i.e., $S'=\{s,r\}\cup S_2 \cup S_3$. After performing the above adjustments, the squares in $S$ remain non-overlapping. Also the squares in $S\setminus S'$ remain unchanged. Thus, by moving $\ell$ on the vertical range $[a-\epsilon, a+\epsilon]$, we obtain a valid solution for the BASP problem. For a given $x\in [a-\epsilon, a+\epsilon]$, the total area of the squares in $S'$ is

$$f(x)=x^2+(|pq|-x)^2+|S_2|\cdot (1-x)^2+\sum_{s_3\in S_3} \mathrm{area}(s_3),$$

where $|S_2|$ denotes the cardinality of $S_2$. As discussed above, the area of $s_3$ is of the form $(c-(1-x))^2$ for some constant $c$. This implies that $f(x)=\alpha x^2+\beta x+\gamma$ for some constants $\alpha>0$, $\beta$, and $\gamma$. This means that $f(x)$ is a convex function on the domain $[a-\epsilon, a+\epsilon]$. Thus, the maximum value of $f(x)$ is attained at an endpoint of the domain, but not at $a$. Therefore, the original solution $S$, for which $\ell$ has $y$-coordinate $a$, cannot be an optimal solution for the BASP problem.
\end{proof}

The set $\mathcal{L}$ contains $O(n)$ lines per point of $P$, and thus, $O(n^2)$ lines in total.
These lines can be computed in $O(n^2)$ time. By Lemma~\ref{discrete-lemma}, for every optimal solution there exists a separating line in $\mathcal{L}$. Therefore, by checking every line $\ell$ in $\mathcal{L}$ and taking the one that maximizes the total area of the two one-sided instances of the problem (one for each side of $\ell$), we can solve the two-sided version of the problem in $O(n^4)$ time. 

\paragraph{Remark.} A restricted version of the BASP problem, where every point of $P$ should be assigned a non-zero square, can be solved in $O(n)$ time for the one-sided case, and in $O(n^2)$ time for the two-sided case. In the one-sided case, we have a constant number of squares/intervals per point because we only need to check its two neighbors. By a similar reason, in the two-sided case we get a constant number of lines per point, and thus, $O(n)$ lines in total.

\section{Conclusion}
\label{sec:conclusion}
In this paper, we considered the anchored rectangle and square packing problems in which all points are on the boundary of the square $Q$. By exploiting the properties of an optimal solution, we gave an optimal linear-time exact algorithm for the rectangle packing problem. Observe that our algorithm covers nearly everything for large $n$ (contrasting with the fraction of $7/12{-}\varepsilon$ achieved in the non-boundary case~\cite{BalasDT16}). For there are (up to rotation) at least $n/2$ points in $R_B\cup P_\mathcal{T}$, which define $n/2+1$ vertical slabs.  Rule (1) or (2b) will consider the narrowest of them as hole, which has area at most $1/(n/2+1)$ if $Q$ has area 1. So, we cover a fraction of $1-O(1/n)$ of $Q$. We also considered the square packing problem when the points on $P$ are on two opposite sides of $Q$, and gave an $O(n^4)$-time algorithm for this problem.

The most interesting open question is to determine the complexity of the BARP or BASP problem for when the points of $P$ can lie in the interior of $Q$. Is it polynomial-time solvable? As a first step, it would be interesting to characterize which polygonal curves on $Q\cup \grid{P}{S}$ could be boundaries of a hole in a solution. Moreover, the complexity of the BASP problem when the points of $P$ are on all four sides of $Q$ remains open.

\subparagraph{Acknowledgement.} The authors thank Paz Carmi and Lu\'{i}s Fernando Schultz Xavier da Silveira for helpful discussions on the problem.

\bibliographystyle{plain}
\bibliography{ref}

\end{document}